\documentclass[lettersize,journal,twoside]{IEEEtran}
\usepackage{algorithm}  
\usepackage{amsmath}  
\usepackage{amssymb}  
\usepackage{enumitem}  
\usepackage{graphicx}  
\usepackage[caption=false,font=footnotesize]{subfig}  
\usepackage{siunitx}  
\usepackage[colorlinks=true,linkcolor=blue,citecolor=blue,urlcolor=black]{hyperref}
\usepackage{cleveref}
\usepackage{mathtools}  
\usepackage{booktabs}
\usepackage[nocompress]{cite}
\usepackage{setspace}

\renewcommand{\IEEEproofindentspace}{\parindent}

\newcommand{\SEL}{\texttt{SEL}}
\newcommand{\CTRL}{\texttt{$\Pi$}}

\newcommand{\vpar}{v^\textnormal{par}}
\newcommand{\Vpar}{\mathcal{V}^\textnormal{par}}
\newcommand{\vnom}{v^\textnormal{nom}}

\newcommand{\vmin}{\underline{v}}
\newcommand{\vmax}{\overline{v}}
\newcommand{\qmin}{\underline{q}}
\newcommand{\qmax}{\overline{q}}

\newcommand{\qlims}{[\qmin, \qmax]}
\newcommand{\vlims}{[\vmin, \vmax]}
\newcommand{\etamax}{\overline\eta}
\DeclareMathOperator*{\mean}{mean}  
\DeclareMathOperator{\vech}{vech}
\DeclareMathOperator{\NCBC}{NCBC}

\newcommand{\R}{\mathbb{R}}  
\newcommand{\X}{\mathcal{X}}  
\newcommand{\Z}{\mathbb{Z}}  
\newcommand{\Sym}{\mathbb{S}}  
\newcommand{\Ncal}{\mathcal{N}}  
\newcommand{\Ecal}{\mathcal{E}}  
\newcommand{\abs}[1]{\left\lvert#1\right\rvert}  
\newcommand{\iid}{\stackrel{\mathrm{iid}}{\sim}}  
\newcommand{\norm}[1]{\left\lVert#1\right\rVert}  
\newcommand{\one}{\mathbf{1}}  
\newcommand{\set}[1]{\left\{#1\right\}}  
\newcommand{\zero}{\mathbf{0}}  

\DeclareMathOperator*{\argmin}{arg\,min}
\DeclareMathOperator{\diam}{diam}  
\DeclareMathOperator{\var}{Var}  

\newtheorem{theorem}{Theorem}
\newtheorem{lemma}{Lemma}
\newtheorem{definition}{Definition}
\newtheorem{assumption}{Assumption}

\crefname{lemma}{Lemma}{Lemmas}
\crefname{definition}{Definition}{Definitions}
\crefname{assumption}{Assumption}{Assumptions}

\begin{document}

\bstctlcite{reference:BSTcontrol}

\title{Online Learning for Robust Voltage Control\\Under Uncertain Grid Topology}

\author{Christopher Yeh,~\IEEEmembership{Graduate Student Member,~IEEE}, Jing Yu,~\IEEEmembership{Graduate Student Member,~IEEE},\\Yuanyuan Shi,~\IEEEmembership{Member,~IEEE}, Adam Wierman,~\IEEEmembership{Senior Member,~IEEE}
\thanks{Manuscript received 17 August 2023; revised 11 January 2024; accepted 15 March 2024. Date of publication 1 April 2024; date of current version 23 August 2024. This work was supported in part by the Caltech Resnick Sustainability Institute; in part by two Caltech/Amazon AWS AI4Science Fellowships; and in part by the National Science Foundation under Grant CNS-2146814, Grant CNS-2106403, Grant CPS-2136197, Grant ECCS-2200692, and Grant NGSDI-2105648. \textit{(Corresponding author: Christopher Yeh.)}}
\thanks{Christopher Yeh, Jing Yu, and Adam Wierman are with the Department of Computing and Mathematical Sciences, California Institute of Technology, Pasadena, CA 91125 USA (email: \url{cyeh@caltech.edu})}
\thanks{Yuanyuan Shi is with the Department of Electrical and Computer Engineering, University of California, San Diego, La Jolla, CA 92093 USA.}
\thanks{Color versions of one or more figures in this article are available at
\url{https://doi.org/10.1109/TSG.2024.3383804}.}
\thanks{Digital Object Identifier 10.1109/TSG.2024.3383804}
}

\markboth{IEEE Transactions on Smart Grid}{Yeh \MakeLowercase{et al.}: Online learning for robust voltage control under uncertain grid topology}

\IEEEpubid{\begin{minipage}{\textwidth}\ \\[12pt] \centering
\copyright~2024 The Authors. This work is licensed under a Creative Commons Attribution 4.0 License.\\
For more information, see \url{https://creativecommons.org/licenses/by/4.0/}
\end{minipage}}


\maketitle

\begin{abstract}
Voltage control generally requires accurate information about the grid's topology in order to guarantee network stability. However, accurate topology identification is challenging for existing methods, especially as the grid is subject to increasingly frequent reconfiguration due to the adoption of renewable energy. In this work, we combine a nested convex body chasing algorithm with a robust predictive controller to achieve provably finite-time convergence to safe voltage limits in the online setting where there is uncertainty in both the network topology as well as load and generation variations. In an online fashion, our algorithm narrows down the set of possible grid models that are consistent with  observations and adjusts reactive power generation accordingly to keep voltages within desired safety limits. Our approach can also incorporate existing partial knowledge of the network to improve voltage control performance. We demonstrate the effectiveness of our approach in a case study on a Southern California Edison 56-bus distribution system. Our experiments show that in practical settings, the controller is indeed able to narrow the set of consistent topologies quickly enough to make control decisions that ensure stability in both linearized and realistic non-linear models of the distribution grid.
\end{abstract}


\section{Introduction}

\IEEEPARstart{O}{perators} of electricity distribution grids must maintain voltages at each bus within certain operating limits, as deviations from such limits may damage electrical equipment and cause power outages \cite{regulator,haes2019survey}. This ``voltage control'' or ``voltage regulation'' problem has been well-studied, \textit{e.g.}, ~\cite{rebours2007survey,zhu2015fast,molzahn2017survey} and the references therein. Voltage control devices and algorithms aim to guarantee grid stability and minimize the costs associated with control inputs. While classic voltage regulation devices such as tap-changing transformers are effective in dealing with \emph{slow} voltage variations ~\cite{senjyu2008optimal,gao2010review}, increasing penetration of renewables leads to faster variations, and a growing body of literature has focused on inverter-based controllers that can respond quickly by adjusting their active and reactive power set-points. Most of these works cast voltage control as an optimization problem and then propose different centralized or decentralized algorithms depending on the communication infrastructure.

Typically, voltage control algorithms assume \textit{exact knowledge} of the underlying grid topology. This includes centralized controllers such as algorithms based on model predictive control (MPC) which optimize control decisions for a short-term horizon. \cite{guo2019mpc} uses MPC to manage distributed generation and energy storage systems, whereas \cite{maharjan2021robust} proposes a robust MPC controller that is robust to uncertainty in the forecasts of future loads and solar generation.

However, the exact grid topology and line parameters are often not known, and using existing voltage control algorithms with incorrect grid information may lead to problems with grid stability~\cite{park2018exact,li2019robust}. For example, parts of the grid may undergo reconfiguration due to load balancing or unplanned maintenance, as frequently as every hour of the day \cite{dehghanian2016probabilistic,li2016determination,deka2017structure,esmaeili2019optimal}. This problem is exacerbated by the increasing integration of distributed energy resources (DERs), such as photovoltaic (PV) and storage devices. Especially in distribution grids, where DERs are not owned or operated by the electricity utility, the grid operator may lack up-to-date information about the grid topology~\cite{liao2015distribution}. While a grid operator can install sensors to help identify the current network topology, unless such sensors are densely deployed (at great cost), uncertainty about the topology remains. Thus, distribution grid operators cannot expect to operate with perfect topology information and the design of voltage control algorithms robust to unknown grid topology is crucial.

\IEEEpubidadjcol
There are several families of existing algorithms that do not require knowing the network topology: decentralized controllers, model-free controllers, and controllers that first try to infer the network topology. While decentralized voltage control algorithms are generally efficient to implement, such controllers lack voltage stability guarantees when the load is time-varying \cite{bolognani2013distributed,li2014realtime,zhu2015fast,tang2019fast,qu2020optimal}. Likewise, model-free controllers based on deep reinforcement learning do not require knowing the network topology, but they generally have no performance or voltage stability guarantees and are therefore not suitable for safety-critical infrastructure \cite{duan2019deep,wang2020data,xu2020optimal,gao2021consensus,sun2021two}. Some recent works~\cite{shi2022stability,feng2022stability,cui2022decentralized} have proposed methods for introducing stability guarantees for model-free deep reinforcement learning approaches. Their main tool is Lyapunov stability theory, from which a structural constraint for stable controllers is derived, and policy optimization with the constraint is performed. However, their stability guarantees are only valid over an infinite time horizon, and achieving good performance with deep reinforcement learning generally requires large amounts of historical training data. In contrast, our proposed framework jointly learns the system model (consistent with data) and stable controller in an online fashion, achieving a finite-mistake guarantee and good performance without relying on historical data.

Another standard approach for handling uncertainty about network topology is to first estimate the topology and line parameters using a form of system identification with data and then apply a standard voltage control algorithm using the identified network topology. There is a growing literature of such data-driven methods, \textit{e.g.}, \cite{kekatos2014grid,liao2015distribution, park2018exact,li2019robust, fabbiani2021identification, brouillon2022bayesian, cavraro2018graph, ardakanian2019identification, xu2019data, nowak2020measurement, deka2023learning, chen2020data}. A common approach is to leverage least squares for system model estimation. The estimation and therefore control guarantees depend on statistical modeling of measurement noise (\textit{e.g.}, Gaussian). In contrast, we leverage online learning in order to be robust against any bounded disturbances, such as modeling errors and adversarial noise. While least squares-based algorithms focus on \textit{asymptotic} estimation convergence, \textit{e.g.} \cite{brouillon2022robust, yu2017patopa}, we present a \textit{finite} mistake guarantee that is crucial for safe \textit{transient} system behavior.

Another prominent approach is to use graphical models for topology reconstruction~\cite{deka2016estimating}, via maximum likelihood methods while enforcing other structural restrictions like low-rank and sparsity. However, these methods that first perform some form of system identification have drawbacks. First, the estimated topology and/or system dynamics may be imperfect~\cite{sharon2012topology}, and applying standard voltage control algorithms using these imperfect estimates may still lead to system instability. Second, these methods either assume access to historical data or require acquiring data online over hundreds of time steps, during which the stability of the system is ignored~\cite{liao2015distribution,deka2016estimating}. In contrast, our proposed approach does not perform system identification separately from control; the joint operation of our robust controller with the system dynamics estimation gives rise to our stability guarantee.

\subsection{Contributions}
We propose a new approach for voltage control over an uncertain grid topology that does not perform system identification and voltage control separately. Instead, our approach robustly learns to stabilize voltage within the desired limits directly, \emph{without prior knowledge of the topology and without needing to precisely learn the topology}. Our approach takes ideas from online nested convex body chasing (NCBC) \cite{argue2019nearly} and robust predictive control and combines them using a new learning framework~\cite{ho2021online} to design a voltage control algorithm. Intuitively, we use a NCBC algorithm to track the set of topologies that are consistent with the observed voltage measurements---as more measurements are taken, the set of consistent topologies shrinks (and so the sets are nested). As these measurements are taken, a form of robust predictive control is used for voltage control, where the robustness guarantee is used to ensure that the uncertainty about the topology can be handled. Our main result (\Cref{thm:main}) provides a finite error stability bound for the overall controller, which is summarized in \Cref{alg:robust_online_volt_control}. This represents the first voltage control algorithm that is provably robust to large uncertainty about network topology.

This paper supersedes the results of the preliminary version of this work~\cite{yeh2022robust} in the following aspects:
\begin{enumerate}
    \item We improve the analysis of \cite{yeh2022robust}, which assumes no uncertainty in the maximum load/generation variability, to both reduce the mistake bound by a factor of 2 and also improve empirical voltage control performance.
    \item We extend our approach to handle uncertainty in the maximum variability of load and generation entities in the grid, and we show that in the limiting case of 0 uncertainty, our result coincides with the improved analysis mentioned in (1).
    \item We perform case studies of the proposed algorithm on the Southern California Edison (SCE) 56-bus distribution system~\cite{farivar2012optimal} with a more realistic nonlinear power flow model with partial control and partial observation. Even though the design of our method is based on a linear approximation to the power flow model, our method still performs well for the nonlinear system.
    \item We demonstrate how to incorporate existing partial knowledge of the grid topology and network line parameters into the algorithm. We show that incorporating such prior knowledge can improve the performance of our algorithm.
\end{enumerate}
\section{Model}
\label{sec:volt_ctrl_background}
We study voltage control on an unknown grid topology. We consider a radial (tree-structured) power distribution network represented as a connected directed graph $G = (\Ncal, \Ecal)$, where $\Ncal = \{0, 1, 2, \dotsc, n\}$ is the set of buses (nodes) and $\Ecal \subset \Ncal \times \Ncal$ is the set of lines (directed edges). Let the network be rooted at bus 0 (the substation or slack bus), and let other buses be branch buses. Let $\mathcal{C} \subseteq \Ncal$ denote the subset of buses with controllable reactive power injection. Because the network is radial and rooted at bus 0, there is a unique path $\mathcal{P}_i$ from bus 0 to any other bus $i$. For branch buses, let $v \in \R^n$ be their squared voltage magnitudes and $p + \mathbf{i} q$ be their complex power injection, where $p \in \R^n$ (units \si{\watt}) is the net active power injection, and $q \in \R^n$ (units \si{\var}) is the net reactive power injection. The DistFlow branch equations~\cite{baran1989optimal} for a distribution grid are as follows, for all $j \in \Ncal$ and $(i,j) \in \Ecal$:
\begin{subequations}
\label{eq:nonlinear-distflow}
\begin{align}
    -p_j &= P_{ij} - r_{ij} {l_{ij}} - \sum_{k: (j, k) \in \Ecal} P_{jk} \label{eq:bfm_p} \\
    -q_j &= Q_{ij} - x_{ij} {l_{ij}} - \sum_{k: (j, k) \in \Ecal} Q_{jk} \label{eq:bfm_q}\\
    v_j &= v_i - 2(r_{ij}P_{ij} + x_{ij} Q_{ij}) + (r_{ij}^2 + x_{ij}^2) l_{ij} \label{eq:bfm_v} \\
    l_{ij} &= \frac{P_{ij}^2 + Q_{ij}^2}{v_i} \label{eq:bfm_l}
\end{align}
\end{subequations}
where $P_{ij}$ and $Q_{ij}$ represent the active power and reactive power flow on line $(i,j)$, and $r_{ij}, x_{ij} > 0$ are the real-valued line resistance and reactance (units \si{\ohm}). \Cref{eq:bfm_p,eq:bfm_q} represent the real and reactive power conservation at bus $j$, and \eqref{eq:bfm_v} represents the voltage drop from bus $i$ to bus $j$.

Assuming the branch power losses ($r_{ij} l_{ij}$, $x_{ij} l_{ij}$) are negligible yields the simplified DistFlow equations \cite{baran1989optimalsizing}, which can be rearranged into
\begin{equation}\label{eq:simplified-distflow}
    v = R^\star p + X^\star q + v^0 \one_n
\end{equation}
where $v^0 \in \R^n$ is the known, constant squared voltage magnitude at the substation, and $R^\star, X^\star \in \Sym^n$ are computed from the network topology and line parameters
\begin{equation}\label{eq:true_RX}
    R^\star_{ij} := 2 \sum_{\mathclap{(h,k) \in \mathcal{P}_i \cap \mathcal{P}_j}} r_{hk},
    \quad
    X^\star_{ij} := 2 \sum_{\mathclap{(h,k) \in \mathcal{P}_i \cap \mathcal{P}_j}} x_{hk},
    \quad
    i,j \in [n]
\end{equation}
with $[n] := \{1, \dotsc, n\}$ \cite{li2014realtime}. ($\Sym^n$ is the set of symmetric $n \times n$ matrices.) $R^\star, X^\star$ are positive definite with nonnegative entries \cite{farivar2013equilibrium}, and the largest entry of each row of these matrices is along the diagonal, since
\begin{equation}\label{eq:diag-largest}
    X^\star_{ij}
    = 2 \sum_{\mathclap{(h,k) \in \mathcal{P}_i \cap \mathcal{P}_j}} x_{hk}
    \leq 2 \sum_{\mathclap{(h,k) \in \mathcal{P}_i}} x_{hk}
    = X^\star_{ii}
\end{equation}
and likewise for $R^\star_{ij} \leq R^\star_{ii}$.

We assume that the active power injection $p$ is exogenous but that reactive power at each bus can be decomposed as $q = q^c + q^e$, where $q^c$ is the ``controllable'' component and $q^e$ is the ``exogenous'' (\textit{i.e.}, uncontrollable) component. Following \cite{li2014realtime}, we define $\vpar = R^\star p + X^\star q^e + v^0 \one_n \in \R^n$ (``par'' stands for ``partial'') representing the exogenous effects on voltage. Then, $v = X^\star q^c + \vpar$, which can be modeled as a discrete-time linear system
\begin{equation}\label{eq:vol_dyn1}
    v(t+1) = X^\star q^c(t) + \vpar(t).
\end{equation}
Substituting $u(t) = q^c(t) - q^c(t-1)$ (change in controllable reactive power injection) and $w(t) = \vpar(t) - \vpar(t-1)$ (change in exogenous noise) yields the linear dynamical system
\begin{equation}\label{eq:vol_dyn2}
    v(t+1) = v(t) + X^\star u(t) + w(t).
\end{equation}

The voltage control problem~\cite{farivar2012optimal} is to drive the squared voltage magnitudes of each bus from an initial state $v(1) \in \R^n$ into a given multi-dimensional interval $\vlims \subset \R^n$; it is possible that $v(1)$ does not start within the interval due to some large initial disturbance. For all $t \geq 2$, the voltage control algorithm aims to maintain $v(t)$ within $\vlims$, ideally as close as possible to a ``nominal'' value $\vnom \in \vlims$, typically $\vnom = (\vmin + \vmax)/2$. The cost for deviating from $\vnom$ is measured by $   \norm{v(t) - \vnom}_{P_v}^2$ for some positive semidefinite matrix $P_v$, where $\norm{x}_A^2 := x^\top A x$.

At each time step, buses may change their reactive power injection $q^c(t)$ in order to regulate the voltage close to $\vnom$. The reactive power injection (including $q^c(0)$) is limited within a given bound $\qlims \subset \R^n$. Buses not in $\mathcal{C}$ do not have any ability to control the reactive power injection: $\forall i \not\in \mathcal{C}.\ \qmin_i = \qmax_i = 0$. We do not place any hard ``ramp constraints'' on $u(t)$. However, we impose a quadratic ramping cost $\norm{u(t)}_{P_u}^2$ where $P_u$ is a positive semidefinite matrix.

In summary, the voltage control problem is to determine an online sequence of reactive power injections $q^c(1), q^c(2), \dotsc$ to drive voltages $v(t)$ to a desired interval $\vlims$ while minimizing voltage violation and control costs $\norm{v(t) - \vnom}_{P_v}^2 + \norm{u(t)}_{P_u}^2$. In this work, we solve the voltage control problem in the setting where $X^\star$ is unknown.

\section{Robust Online Voltage Control}

In this section we introduce our robust online voltage control algorithm (\Cref{alg:robust_online_volt_control}) and its performance bound (\Cref{thm:main}), which is the main result of this paper.

\begin{figure}[t]
    \centering
    \includegraphics[width=\columnwidth]{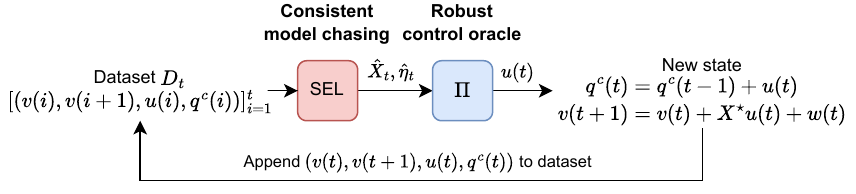}
    \caption{Online robust voltage control framework}
    \label{fig:orc_framework}
\end{figure}

\subsection{Algorithm}

As shown in \Cref{fig:orc_framework}, the algorithm has two main components: a consistent model chasing algorithm \SEL{} (\Cref{alg:robust_online_volt_control}, step 2)
and a robust control oracle \CTRL{} (\Cref{alg:robust_online_volt_control}, step 3). \SEL{} and \CTRL{} are combined by adapting ideas from \cite{ho2021online}.

The model chasing algorithm \SEL{} selects a consistent model for the robust control oracle \CTRL{} out of all plausible models that are consistent with the online observations and prior knowledge of the grid. The selection may use any competitive NCBC algorithm, which is the online problem of choosing a sequence of points within sequentially nested convex sets, with the aim of minimizing the sum of distances between the chosen points \cite{argue2019nearly}. In our experiments, we use a simple projection-based NCBC algorithm, detailed in \Cref{sec:experiments}.

The robust control oracle \CTRL{} is a novel robust predictive controller (\Cref{thm:oracle}). The robustness guarantee of \CTRL{} is necessary for the analysis which integrates \SEL{} with \CTRL{} to provide the finite mistake guarantee of the overall algorithm. We remark that other choices for either component are possible, as long as they provide the guarantees needed in the analysis in \Cref{sec:proof}.

Intuitively, \SEL{} and \CTRL{} are combined in a way such that \SEL{} always reduces the uncertainty about the unknown model whenever \CTRL{} outputs an action that causes a voltage limit violation. This means that \CTRL{} cannot take too many ``bad'' actions before the system uncertainty is small.

\subsection{Assumptions}

Before presenting the main results, we introduce three assumptions that underlie our analysis and discuss why they are both needed and practical.

\begin{assumption}
\label{as:bounded_w}
The change in noise is bounded as
\begin{equation*}
    \forall t:\quad \norm{w(t)}_\infty \leq \eta^\star,
\end{equation*}
where $w(t) = \vpar(t) - \vpar(t-1)$. $\eta^\star \in [0, \etamax]$ is a constant (possibly unknown), while $\etamax$ is a known upper-bound.
\end{assumption}

This first assumption is standard and bounds the noise in the dynamics. It represents realistic behavior in power systems where the active and exogenous reactive power injections do not vary dramatically between time steps, as can be seen by expanding $w(t)$:
\begin{align*}
    w(t)
    &= \vpar(t) - \vpar(t-1) \\
    &= R^\star(p(t) - p(t-1)) + X^\star(q^e(t) - q^e(t-1)).
\end{align*}
For example, if the net active and exogenous reactive power injection is the same at time steps $t$ and $t-1$, then $w(t) = 0$.

An unknown $\eta^\star$ indicates uncertainty in the maximum variability of the exogenous power injections. Unlike \cite{yeh2022robust} which assumes a fixed $\eta$, our inclusion of both an unknown $\eta^\star$ and a known upper-bound $\etamax$ allows more flexibility in our algorithmic design and the incorporation of prior knowledge.

\begin{assumption}
\label{as:bounded_X}
The true model $X^\star$ lies within a known compact, convex uncertainty set $\X \subset \Sym_+^n \cap \R_+^{n \times n}$. ($\Sym_+^n$ is the set of $n \times n$ positive semidefinite matrices, and $\R^{n \times n}_+$ is the set of $n \times n$ matrices with nonnegative entries.)
\end{assumption}

Our second assumption bounds the uncertainty about the network topology and line parameters. It ensures that the unknown true model parameters $X^\star$ belong to a compact, convex set $\X$, which is a minimal assumption necessary for proving an analytic guarantee. $P_1 = \X \times [0,\etamax]$ forms the initial ``consistent set'' (see \Cref{def:consistent}) for our consistent model chasing algorithm \SEL.

This assumption is realistic, as a grid operator should have at least some prior knowledge about the distribution grid topology and the range of possible line parameters, even if they do not have the exact values. In cases where the grid has multiple possible topologies due to switches, $\X$ could be set to the convex hull of the corresponding $X$ matrices.

\begin{definition}[$\norm{\cdot}_\triangle$ and $\norm{\cdot}_{\triangle,\delta}$]
For any matrix $X \in \Sym^n$ and scalars $\eta, \delta \geq 0$, define
\begin{align*}
    \norm{X}_\triangle
    &:= \norm{\vech(X)}_2
    = \sqrt{\sum_{i=1}^n \sum_{j=i}^n X_{ij}^2}
    \\
    \norm{(X,\eta)}_{\triangle,\delta}
    &:= \sqrt{\delta^2\eta^2 + \sum_{i=1}^n \sum_{j=i}^n X_{ij}^2}
    = \sqrt{\delta^2\eta^2 + \norm{X}_\triangle^2}.
\end{align*}
For any sets $\X \subseteq \Sym^n$ and $A \subseteq \R$, we define diameters $\diam(\X)$ and $\diam(\X \times A)$ with respect to the norms $\norm{\cdot}_\triangle$ and $\norm{\cdot}_{\triangle,\delta}$, respectively.
\end{definition}

These norms isometrically map our parameter space to Euclidean space, enabling us to take advantage of known results on NCBC within Euclidean space. For the norm $\norm{\cdot}_{\triangle,\delta}$, the hyperparameter $\delta$ trades off the weight between $X$ and $\eta$ in the norm. The choice of $\delta$ is discussed in \Cref{sec:experiments}.

In practice, we consider uncertainty sets of the form
\begin{equation*}
    \X_\alpha
    = \set{
        X \in \Sym_+^n \cap \R_+^{n \times n}
        \ \middle| \
        \begin{aligned}
        & \|X - X^\star\|_\triangle \leq \alpha \norm{X^\star}_\triangle, \\
        & \forall i,j \in [n]: X_{ij} \leq X_{ii}
        \end{aligned}
    }
\end{equation*}
with $\diam(\X_\alpha) = 2 \alpha \norm{X^\star}_\triangle$. A larger $\alpha$ yields a larger uncertainty set. From \Cref{sec:volt_ctrl_background} (\textit{e.g.}, \eqref{eq:diag-largest}), we know that $X^\star \in \X_\alpha$.

Furthermore, we can incorporate partial knowledge we may have of the network topology and/or line parameters by adding constraints to the description of $\X$. For example, if we know that the lowest common ancestor between buses $i,j$ in the network is bus $k$, then we can add the following constraint on $X$, which is a consequence of \eqref{eq:true_RX}:
\begin{equation}\label{eq:known-topology-constraint}
    X_{ij} =
    \begin{cases}
    0, & k = 0 \\
    X_{kk}, & \text{otherwise}.
    \end{cases}
\end{equation}
If we additionally know the values for some line parameters $x_{ij}$, we may be able to further constrain some entries of $X$, again by applying \eqref{eq:true_RX}.

\begin{assumption}
\label{as:bounded_vpar}
There exists a compact, convex set $\Vpar \subset \R^n$ such that $\forall t \geq 0:\ \vpar(t) \in \Vpar$. Furthermore, for some known $\epsilon > 0$,
\begin{align*}
& \forall \vpar \in \Vpar,\, X \in \X. \\
& \exists q^c \in \qlims \text{ s.t. } X q^c + \vpar \in [\vmin + (\etamax+\epsilon)\one,\, \vmax - (\etamax+\epsilon)\one].
\end{align*}
\end{assumption}

Our final assumption is about the existence of feasible control actions for the robust control oracle. This assumption can be interpreted as either a bound on the noise, or a requirement that the controllable reactive power injection be flexible enough to satisfy the demand of any admissible noise. It represents the reasonable assumption that a grid operator should have installed enough controllable reactive power injection capability to perform voltage control. Intuitively, the $\etamax$ padding is required for robustness to the noise $w(t)$, while the $\epsilon$ padding is required for robustness to model uncertainty (\textit{i.e.}, uncertainty about $X^\star$).

\begin{algorithm}[t!]
\caption{Online Robust Voltage Controller}
\label{alg:robust_online_volt_control}
Inputs
\begin{itemize}
    \item desired nominal squared voltage magnitude: $\vnom \in \R^n$
    \item limits on the squared voltage magnitude: $\vlims \subset \R^n$
    \item limits on the reactive power injection: $\qlims \subset \R^n$
    \item initial state: $v(1), q^c(0) \in \R^n$
    \item state and action cost matrices: $P_v, P_u \in \Sym^n_+$
    \item compact convex uncertainty set for the model parameter: $\X \subset \Sym_+^n \cap \R^{n \times n}_+$
    \item compact convex uncertainty set for exogenous voltage quantities: $\Vpar \subset \R^n$
    \item upper bound for noise: $\etamax > 0$
    \item robustness padding: $\epsilon > 0$
    \item weight for slack variable: $\beta > 0$
    \item weight for noise accuracy: $\delta > 0$
\end{itemize}
Procedure
\begin{enumerate}[leftmargin=*]
    \item Initialize an empty trajectory $D_0 = [\,]$. Set $t = 1$.
    \item Query the model chasing algorithm for a new consistent parameter estimate: $(\hat{X}_t, \hat\eta_t) \leftarrow \SEL[D_t]$.
    \begin{subequations}\label{eq:cmc}
    \begin{align}\hspace{-1mm}
        & \SEL[D_t] := \NCBC(P_t,\hat{X}_{t-1},\hat\eta_{t-1}) \label{eq:ncbc} \\
        & P_t := \left\{
        (\hat{X}, \hat\eta)
        \ \middle|\
        \begin{aligned}
        & \hat{X} \in \X,\ \hat\eta \in [0,\etamax] \\
        & \forall (v_i, v_{i+1}, u_i, q^c_i) \in D_t: \\
        &\ \ \|v_{i+1} - v_i - \hat{X}_t u_i\|_\infty \leq \hat\eta \\
        &\ \ v_{i+1} - \hat{X} q_i^c \in \Vpar
        \end{aligned}
        \right\} \label{eq:alg1_consistent_set}
    \end{align}
    \end{subequations}

    \item Query the robust control oracle for the next control action: $u(t) \leftarrow \CTRL_{\hat{X}_t,\hat\eta_t}(v(t))$.
    \begin{subequations}\label{eq:oracle_slack}
    \begin{align}
       \CTRL_{\hat{X}_t,\hat\eta_t}:\ \min_{u, \xi}\
       & \norm{\hat{v}' - \vnom}_{P_v}^2 + \norm{u}_{P_u}^2 + \beta \xi ^2
            \label{eq:oracle_slack_cost} \\
        \text{s.t.}\
        & u \in \R^n,\ \xi \in \R_+ \\
        & \qmin \preceq q^c(t-1) + u \preceq \qmax
            \label{eq:oracle_slack_constr_u} \\
        & \hat{v}' = v(t) + \hat{X}_t u
            \label{eq:oracle_slack_constr_dynamics} \\
        & k = \hat\eta_t + \rho \left(\frac{1}{\delta} + \norm{u}_2\right)
            \label{eq:oracle_slack_constr_buffer} \\
        & \vmin + (k-\xi) \one \preceq \hat{v}' \preceq \vmax - (k-\xi) \one
            \label{eq:oracle_slack_constr_robust}
    \end{align}
    \end{subequations}

    where $\rho = \delta\epsilon / (1 + \delta \|\qmax-\qmin\|_2)$.

    \item Apply the control action $u(t)$. Observe the system transition to $v(t+1) = v(t) + X^\star u(t) + w(t)$ and $q^c(t) = q^c(t-1) + u(t)$.
    \item Append $(v(t), v(t+1), u(t),q^c(t))$ to the trajectory:
    \[
        D_t = \left[ (v(i), v(i+1), u(i), q^c(i)) \right]_{i=1}^t.
    \]
    \item Increment $t \leftarrow t+1$. Repeat from Step (2).
\end{enumerate}
\end{algorithm}

\subsection{Main result}

We now state our main result, which is a finite-error bound for \Cref{alg:robust_online_volt_control}.

\begin{theorem}[Main Result]
\label{thm:main}
Under \Cref{as:bounded_w,as:bounded_X,as:bounded_vpar}, \Cref{alg:robust_online_volt_control} ensures that the voltage limits will be violated at most $\frac{2 \gamma(m)}{\rho} \diam(\X \times [0,\etamax]) + 1$ times, where $\rho = \frac{\delta\epsilon}{1+\delta\|\qmax-\qmin\|_2}$ and $\gamma(m)$ is the competitive ratio of the NCBC algorithm in $m$-dimensional Euclidean space, where $m = 1 + \frac{n(n+1)}{2}$.

If $\eta^\star$ is known, then the voltage limits will be violated at most $\frac{2 \gamma(m)}{\rho} \diam(\X) + 1$ times, where $\rho = \frac{\epsilon}{\|\qmax-\qmin\|_2}$ and $m = \frac{n(n+1)}{2}$.
\end{theorem}

To the best of our knowledge, this result is the first provable stability bound for voltage control in a setting where the network topology is unknown. It highlights that \Cref{alg:robust_online_volt_control} can ensure stability even after \emph{unknown} changes to the network topology, \textit{e.g.}, due to maintenance, failures, etc., without the need to perform system identification while remaining robust to any bounded and potentially adversarial perturbations satisfying \Cref{as:bounded_w,as:bounded_vpar}.

Intuitively, this result guarantees that the model chasing algorithm \SEL{} will learn a ``good enough'' model for control quickly. When the robust controller \CTRL{} makes a mistake, the model chasing algorithm will learn from that mistake and significantly reduce the set of consistent models. Because the initial set of consistent models is bounded, and this set shrinks a significant amount after each mistake, the total number of mistakes is bounded. Note that this finite mistake bound implies finite-time convergence to safe voltage limits without an explicit finite-time bound.

To interpret the error bounds in \Cref{thm:main}, we notice that they are proportional to the diameter of the parameter space and the competitive ratio $\gamma(m)$ of the NCBC algorithm, and inversely proportional to the oracle robustness margin $\rho$. Because of computational tractability concerns, our experiments implement \SEL{} with a greedy projection-based NCBC algorithm with $\gamma_\textnormal{proj}(m) = \pi (m-1) m^{m/2}$ \cite{argue2019nearly}, rather than the state-of-the-art Steiner point method which can achieve $\gamma_\textnormal{Steiner}(m) = m/2$~\cite{bubeck2020chasing}. As our case studies show, in practice the projection-based NCBC algorithm performs much better than the worst-case bound. We note that any other NCBC algorithm with a finite competitive ratio can be used in \eqref{eq:ncbc} in \Cref{alg:robust_online_volt_control}. Investigating whether widely-used estimation methods, like least squares, have a finite competitive ratio would be an interesting avenue for future research.

Note that for \Cref{thm:main} to hold, the optimization problem for the robust control oracle \CTRL{} should first be solved without the slack variable $\xi$ in \Cref{alg:robust_online_volt_control}. This ensures that if $(\hat{X}_t,\hat\eta_t)$ is sufficiently close enough to the true model, then the algorithm will not make a mistake. In the case that \CTRL{} is infeasible initially (\textit{e.g.}, when the initial model estimate is far from the true model), it should be solved again with a slack variable, which ensures feasibility. However, solving \CTRL{} twice is unnecessary in practice, and so we have written \Cref{alg:robust_online_volt_control} to reflect its practical implementation.

We outline a proof of \Cref{thm:main} in the next section. We want to highlight one piece of that proof that is of independent interest. In particular, a major step in the proof is to provide a feasibility guarantee for the robust control oracle component \CTRL{} of the algorithm, which is done in \Cref{thm:oracle}.

\section{Proofs}
\label{sec:proof}

We now prove our main result \Cref{thm:main}. Our proof builds on and adapts the approach of~\cite{ho2021online}, which outlines a general framework for integrating model chasing and robust control. To explain the general framework, we first consider a discrete-time nonlinear dynamical system
\[
    x_{t+1} = f_*(x_t, u_t) + w_t,
    \qquad
    x_0 \text{ given},
    \qquad
    (f_*, \mathbf{w}) \in \mathcal{F},
\]
where $x \in \mathcal{S} \subseteq \R^n$ is the system state and $u \in \mathcal{U} \subseteq \R^m$ is the control input. The unknown function $f_*$ and disturbance sequence $\mathbf{w} \in \ell^\infty(\Z_+; \R^n)$ belong to an uncertainty set $\mathcal{F}$, and the disturbance is bounded as $\norm{\mathbf{w}}_\infty \leq \etamax$. Assume that $\mathcal{F}$ has a \emph{compact parametrization} $(\mathbb{T}, \mathsf{K}, d)$, where $\mathbb{T}: \mathsf{K} \to \wp(\mathcal{F})$ is a mapping from a parameter space $\mathsf{K}$ to a set of functions and disturbances such that $\mathcal{F} \subseteq \bigcup_{\theta \in \mathsf{K}} \mathbb{T}[\theta]$. $\wp(\mathcal{F})$ denotes the powerset of $\mathcal{F}$. Let $d$ denote a metric on $\mathsf{K}$, so $(\mathsf{K}, d)$ is a compact metric space.

The control objective is specified as a sequence of indicator ``goal" functions $\mathcal{G} = (\mathcal{G}_0, \mathcal{G}_1, \dotsc)$. Each $\mathcal{G}_t: \mathcal{X} \times \mathcal{U} \to \{0,1\}$ encodes a desired condition per time step $t$:
\[
    \mathcal{G}_t(x_t, u_t) = \one[\text{$x_t$, $u_t$ violate desired condition at time $t$}].
\]
The main result of \cite{ho2021online} specifies a set of sufficient conditions for a finite-mistake guarantee---\textit{i.e.}, $\sum_{t=0}^\infty \mathcal{G}_t(x_t, u_t) < \infty$. These conditions decouple online robust control into separate online learning and robust control components. The online learning component requires a consistent model chasing algorithm \SEL, which takes as input the current observed trajectory $D_t = [(x_i, x_{i+1}, u_i)]_{i=1}^t$ and outputs an estimated parameter $\theta_t \in \mathsf{K}$ which must be \emph{consistent} with $D_t$.

\begin{definition}[Consistent Parameter]
\label{def:consistent}
We say $\theta \in \mathsf{K}$ is consistent with $D_t$ if there exists $(f,\mathbf{w}) \in \mathbb{T}[\theta]$ such that
\[
    \forall (x_t, x_{t+1}, u_t) \in D_t:\ x_{t+1} = f(x_t,u_t) + w_t.
\]
\end{definition}
Let $P_t$ denote the set of all parameters consistent with $D_t$; $P_t$ is called the \emph{consistent set}. We say \SEL{} is $\gamma$-competitive if $\sum_{t=1}^{\infty} d(\theta_t, \theta_{t-1}) \leq \gamma \max_{\theta \in \mathsf{K}} d(P_{\infty}, \theta)$ holds for a fixed constant $\gamma>0$, which we call the \emph{competitive ratio}.

The robust control component requires a control oracle \CTRL, which given the current state $x_t$ and a parameter $\theta_t$, outputs a control action $u_t = \CTRL_{\theta_t}(x_t)$ that is robust for all systems that are close to $\theta_t$. In particular, we call a control oracle \emph{$\rho$-robust} for control objective $\mathcal{G}$, if all trajectories in $S^\Pi[\rho; \theta]$ achieve $\mathcal{G}$ after finitely many mistakes. $S^\Pi[\rho; \theta]$ is defined as the set of all possible trajectories generated by $\CTRL_{\hat\theta}$ for all $\hat\theta$ such that $d(\theta, \hat\theta) \leq \rho$:
\[
    S^\Pi[\rho; \theta] = \left\{
        \begin{aligned}
            & D_\infty = [(x_t, x_{t+1}, u_t)]_{t=1}^\infty: \\
            &\quad u_t = \Pi_{\hat{\theta}}(x_t) \\
            &\quad x_{t+1} = f(x_t, u_t) + w_t
        \end{aligned}
        \,\middle|\,
        \begin{aligned}
        & (f, \mathbf{w}) \in \mathbb{T}[\theta], \\
        & d(\hat{\theta}, \theta) \leq \rho
        \end{aligned}
    \right\}
\]
Due to the page limit, we refer readers to~\cite{ho2021online} for a more detailed discussion of consistent model chasing algorithms and $\rho$-robust control oracles. As a summary, if \SEL{} chases consistent models and \CTRL{} is a robust oracle for $\mathcal{G}$, then the resulting $A_\CTRL(\SEL)$ algorithm achieves a finite mistake guarantee, which is stated in the following.

\begin{theorem} \label{thm:ho}
\cite[Theorem 2.5]{ho2021online}
Assume that \SEL{} chases consistent models and \CTRL{} is a robust oracle for objective $\mathcal{G}$. Then for any starting point $x_0$ and trajectory $[(x_t, u_t)]_{t=0}^\infty$ generated by $\mathcal{A}_\CTRL(\SEL)$ (illustrated in \Cref{fig:orc_framework}), the following mistake guarantees hold: (i) If \CTRL{} is robust, then $\sum_{t=0}^\infty \mathcal{G}_t(x_t, u_t) < \infty$; (ii) If \CTRL{} is uniformly $\rho$-robust and \SEL{} is $\gamma$-competitive, then
\[
    \sum_{t=0}^\infty \mathcal{G}_t(x_t, u_t) < \max\left\{1, M_{\rho}^\CTRL\right\}\left(\frac{2\gamma}{\rho} \diam(\mathsf{K}) +1\right)
\]
where $M_\rho^\CTRL$ denotes the worst case total mistakes of the $\rho$-robust control oracle \CTRL.
\end{theorem}

To apply \Cref{thm:ho} to prove \Cref{thm:main}, we need to prove that (i) the proposed algorithm \eqref{eq:cmc} chases consistent models and has a bounded competitive ratio, and (ii) the proposed robust algorithm in \eqref{eq:oracle} is a $\rho$-robust control oracle, for bounded disturbance in the system topology. In particular, the correspondence of the definitions is as follows. We have $\theta = (X,\eta)$, and
\begin{align*}
    & \mathsf{K} = \X \times [0,\etamax], \quad v(1), q^c(0) \text{ given} \\
    & d((X,\eta), (X',\eta')) = \|(X,\eta) - (X',\eta')\|_{\triangle,\delta} \\
    & \mathbb{T}[(X,\eta)] = \set{
        (f, \mathbf{w}) \ \middle| \
        \begin{aligned}
        & f(v, u) = v + X u,\
        \norm{\mathbf{w}}_\infty \leq \eta, \\
        & \forall t \geq 0:\
        \widehat{\vpar_0} + \sum_{\tau=1}^t w(\tau) \in \Vpar \\
        & \text{where } \widehat{\vpar_0} := v(1) - X q^c(0)
        \end{aligned}
    } \\
    & \mathcal{F} = \bigcup_{(X,\eta) \in \X \times [0,\etamax]} T[(X,\eta)] \\
    & \mathcal{G}_t(v(t)) = \one[v(t) \in \vlims].
\end{align*}

We begin by proving that the set $P_t$ defined in \eqref{eq:alg1_consistent_set} in     \Cref{alg:robust_online_volt_control} is consistent with the trajectory $D_t$.

\begin{lemma}[\SEL{} is consistent]
Suppose $D_T$ is a trajectory of voltage measurements and control actions taken up to time $T$:
\[
    D_T = \left[ (v(t), v(t+1), u(t), q^c(t)) \right]_{t=1}^T.
\]
The set
\begin{equation}
\label{eq:consistent-set}
    \hspace{-0.6em}
    P_T := \set{
        (\hat{X}, \hat\eta)
        \ \middle| \
        \begin{aligned}
        & \hat{X} \in \X,\ \hat\eta \in [0,\etamax], \\
        & \forall (v(t), v({t+1}), u(t), q^c(t)) \in D_T: \\
        &\quad \|v({t+1}) - v(t) - \hat{X} u(t)\|_\infty \leq \hat\eta \\
        &\quad v({t+1}) - \hat{X} q^c(t) \in \Vpar
        \end{aligned}
    }
    \hspace{-0.2em}
\end{equation}
is a consistent set for $D_T$, \textit{i.e.}, $(\hat{X}, \hat\eta)$ is consistent (\Cref{def:consistent}) if and only if $(\hat{X}, \hat\eta) \in P_T$.
\end{lemma}
\begin{IEEEproof}
Consider any $(\hat{X},\hat\eta) \in P_T$. For $t \in [T]$, define
\[
    \hat{f}(v,u) := v + \hat{X} u,
    \qquad
    \hat{w}(t) := v({t+1}) - v(t) - \hat{X} u(t)
\]
so $\|\hat{w}(t)\| \leq \hat\eta$ and $v(t+1) = \hat{f}(v(t), u(t)) + \hat{w}(t)$. Define $\widehat{\vpar_0} := v(1) - \hat{X} q^c(0)$, so for all $t \geq 0$,
\[
    \widehat{\vpar_0} + \sum_{\tau=1}^t w(\tau)
    = v(t+1) - \hat{X} q^c(t)
    \in \Vpar.
\]
Thus, $(\hat{f},\hat{w}) \in \mathbb{T}[(\hat{X}, \hat\eta)]$, so $(\hat{X},\hat\eta)$ is consistent with $D_T$.

Conversely, suppose $(\hat{X},\hat\eta)$ is consistent with $D_T$, which implies the existence of $\hat{f}(v,u) := v + \hat{X} u$ and $\mathbf{\hat{w}}$ satisfying $\norm{\mathbf{\hat{w}}}_\infty \leq \hat\eta$ such that $v(t+1) = \hat{f}(v(t),u(t)) + \hat{w}(t)$. Rearranging yields $\hat{w}(t) = v(t+1) - v(t) - \hat{X} u(t)$, so $(\hat{X},\hat\eta)$ satisfies the norm constraint in \eqref{eq:consistent-set}. Now define
\[
    \forall t \geq 0:\
    \widehat\vpar(t)
    := v(t+1) - \hat{X} q^c(t)
    = \widehat{\vpar_0} + \sum_{\tau=1}^t \hat{w}(t)
\]
so $\widehat\vpar(t) \in \Vpar$ satisfies the remaining constraint in \eqref{eq:consistent-set}.
\end{IEEEproof}

Observe that each $P_t$ is a closed, bounded, and convex set. Furthermore, $P_t$ is non-empty, since $(X^\star,\eta^\star) \in P_t$. Intuitively, $P_t$ is the smallest set containing all parameters that could generate the observed trajectory $D_t$ along with a corresponding admissible sequence of noise compatible with \Cref{as:bounded_w,as:bounded_X,as:bounded_vpar}.

The consistent sets are nested $P_t \subseteq P_{t-1}$, and we use our particular choice of norm $\norm{\cdot}_{\triangle,\delta}$ to establish a linear bijection between $(\Sym^n \times \R, \norm{\cdot}_{\triangle,\delta})$ and Euclidean space $(\R^m, \norm{\cdot}_2)$. This allows us to take advantage of any $\gamma(m)$-competitive NCBC algorithm in Euclidean space \cite{argue2019nearly,bubeck2020chasing}, where $m$ is the dimension of the space, to prove that \SEL{} is $\gamma(m)$-competitive. This is formalized in the following lemma.

\begin{lemma}[\SEL{} is competitive]
\label{thm:ncbc}
If the NCBC algorithm used in \SEL{} has competitive ratio $\gamma(m)$, then \SEL{} is $\gamma(m)$-competitive.
\end{lemma}
\begin{IEEEproof}
The proof is similar to \cite[Lemma 2]{yeh2022robust}, except that learning $\hat\eta$ adds an additional dimension to the parameter space. That is, there exists a norm-preserving linear bijection between $(\Sym^n \times \R, \norm{\cdot}_{\triangle,\delta})$ and Euclidean space $(\R^m, \norm{\cdot}_2)$.
\end{IEEEproof}

Finally, we show that our controller \CTRL{} is $\rho$-robust. In particular, we prove that $\CTRL_{\hat{X}}$ makes no mistakes ($M_\rho^\Pi = 0$) given consistent parameters $(\hat{X},\hat\eta) \in P_t$.

\begin{theorem}[\CTRL{} is $\rho$-robust]
\label{thm:oracle}
Under \Cref{as:bounded_w,as:bounded_X,as:bounded_vpar}, suppose $(\hat{X},\hat\eta) \in P_t$, where $P_t$ is given in \eqref{eq:consistent-set} for $t\geq 1$. Define $\rho = \frac{\delta\epsilon}{1+\delta\|\qmax-\qmin\|_2}$. Then, the following optimization problem is feasible:
\begin{subequations}
\label{eq:oracle}
\begin{align}
    \min_{u \in \R^n}\quad & \norm{\hat{v}' - \vnom}_{P_v}^2 + \norm{u}_{P_u}^2 \label{eq:oracle_cost} \\
    \text{s.t.}\quad
    & \qmin \preceq q^c(t-1) + u \preceq \qmax
        \label{eq:oracle_constr_u} \\
    & \hat{v}' = v(t) + \hat{X}u
        \label{eq:oracle_constr_dynamics} \\
    & k = \hat\eta + \rho \left( \frac{1}{\delta} + \norm{u}_2 \right)
        \label{eq:oracle_constr_buffer} \\
    & \vmin + k \one \preceq \hat{v}' \preceq \vmax - k \one.
        \label{eq:oracle_constr_robust}
\end{align}
\end{subequations}
Further, the solution of \eqref{eq:oracle}, $u(t)$, guarantees voltage stability for all $(X,\eta) \in \X \times [0,\etamax]$ such that $\|(X,\eta)-(\hat{X},\hat\eta)\|_{\triangle,\delta} \leq \rho$. That is, $v(t) + X u(t) + w(t)\in \vlims$ for all $w(t)$ such that $\norm{w(t)}_\infty \leq \eta$.
\end{theorem}

Observe that \eqref{eq:oracle} corresponds to \eqref{eq:oracle_slack} in \Cref{alg:robust_online_volt_control} with the slack variable set to zero. We note that the robustness margin $\rho$ decreases as $\qlims$ increase. The intuitive reason is that the voltage is more sensitive to changes in $\hat{X}$ when the range of possible $u$'s expands. Therefore, a fixed voltage buffer of $\epsilon$ in constraints \eqref{eq:oracle_slack_constr_buffer} and \eqref{eq:oracle_constr_buffer} affords less robustness to changes in $\hat{X}$ as $\qlims$ gets larger.

\begin{IEEEproof}[Proof of \Cref{thm:oracle}]
First, we will show that the following two conditions are sufficient for feasibility of the optimization problem and $\rho$-robustness for the solution $u$.
\begin{itemize}[nosep]
    \item Feasibility: $k \leq \etamax + \epsilon$
    \item Robustness: $k \geq \hat\eta + \rho\sqrt{\frac{1}{\delta^2} + \norm{u}_2^2}$
\end{itemize}
Then, we will show that our choices of $k$ and $\rho$ satisfy these sufficient conditions.

To derive the sufficient condition for feasibility, define
\[
    \widehat\vpar(t-1) := v(t) - \hat{X} q^c(t-1)
\]
as the conjectured noise when we assume the underlying parameter is $\hat{X}$. Since $\hat{X} \in P_t$ and $P_t \subseteq P_{t-1}$, we have $\widehat\vpar(t-1) \in \Vpar$. Then, by \Cref{as:bounded_vpar}, there exists $q^c \in \qlims$ such that
\[
    \vmin + (\etamax+\epsilon)\one
    \preceq \hat{X} q^c + \widehat\vpar(t-1)
    \preceq \vmax - (\etamax+\epsilon)\one.
\]
Set $u = q^c - q^c(t-1)$ (which satisfies \eqref{eq:oracle_constr_u}) and define
\begin{align*}
    \hat{v}'(u)
    &:= v(t) + \hat{X} u
    = v(t) + \hat{X}[q^c - q^c(t-1)] \\
    &= \hat{X} q^c + \widehat\vpar(t-1).
\end{align*}
Recalling \eqref{eq:vol_dyn2}, we can interpret $\hat{v}'(u)$ as the one-step voltage prediction (without disturbance) under the model $\hat{X}$ given control action $u$ and the current voltage $v(t)$. We thus have
\[
    \vmin + (\etamax+\epsilon)\one
    \preceq \hat{v}'(u)
    \preceq \vmax - (\etamax+\epsilon)\one.
\]
Therefore, as long as $k \leq \etamax + \epsilon$, $u$ will satisfy constraint \eqref{eq:oracle_constr_robust}.

Next, we derive the sufficient condition for robustness. Let $u$ be a solution of \eqref{eq:oracle}, so it satisfies \eqref{eq:oracle_constr_robust}. Let $(X,\eta) \in \X \times [0,\etamax]$ be arbitrary parameters satisfying $\|(X,\eta)-(\hat{X},\hat\eta)\|_{\triangle,\delta} \leq \rho$. Define $\rho_X := \|X-\hat{X}\|_\triangle$. By \Cref{lemma:tri_norm},
\begin{equation}\label{eq:xhat_x}
    -\rho_X \norm{u}_2 \one \preceq (X - \hat{X}) u \preceq \rho_X \norm{u}_2 \one.
\end{equation}
Furthermore, suppose
\begin{equation}\label{eq:bound_wt}
    -\eta\one \preceq w(t) \preceq \eta\one.
\end{equation}
Adding together the 3 inequalities \eqref{eq:oracle_constr_robust}, \eqref{eq:xhat_x}, \eqref{eq:bound_wt} yields
\begin{align*}
    \vmin + (k - \rho_X \norm{u}_2 - \eta)\one
    &\preceq v(t) + Xu + w(t) \\
    &\preceq \vmax - (k - \rho_X \norm{u}_2 - \eta)\one.
\end{align*}
Clearly, if $k - \rho_X \norm{u}_2 - \eta \geq 0$, then the desired robustness condition is satisfied. Since
\[
    \|(X,\eta)-(\hat{X},\hat\eta)\|_{\triangle,\delta}^2
    = \rho_X^2 + \delta^2 \abs{\eta-\hat\eta}^2
    \leq \rho^2,
\]
we have $\abs{\eta-\hat\eta} \leq \frac{1}{\delta} \sqrt{\rho^2 - \rho_X^2}$. This implies $\eta \leq \hat\eta + \frac{1}{\delta} \sqrt{\rho^2 - \rho_X^2}$. Therefore, we can express the robustness condition in terms of $\hat\eta$:
\[
    k \geq \hat\eta + \frac{1}{\delta} \sqrt{\rho^2 - \rho_X^2} + \rho_X \norm{u}_2
    =: f(\rho_X).
\]
For $\rho > 0$, $f(\rho_X)$ is strictly concave and twice-differentiable and therefore achieves its maximum when $f'(\rho_X) = 0$. This maximum value is
$
    \hat\eta + \rho \sqrt{\frac{1}{\delta^2} + \norm{u}_2^2}
$.
Thus, if $k$ is at least this value, then we achieve robustness.

Finally, we show that our choices of $k$ and $\rho$ satisfy the sufficient conditions. Since $a + b \geq \sqrt{a^2+b^2}$ for all $a,b \geq 0$, our choice of $k$ satisfies the robustness condition:
\[
    k = \hat\eta + \rho \left( \frac{1}{\delta} + \norm{u}_2 \right)
    \geq \hat\eta + \rho \sqrt{\frac{1}{\delta^2} + \norm{u}_2^2}.
\]
Note that while setting $k = \hat\eta + \rho \sqrt{\frac{1}{\delta^2} + \norm{u}_2^2}$ would also satisfy the robustness condition, this expression would make \eqref{eq:oracle} a nonconvex optimization problem.

The remaining step is to satisfy the feasibility condition. We must choose $\rho$ such that
$
    \hat\eta + \rho \left( \frac{1}{\delta} + \norm{u}_2 \right)
    \leq \etamax + \epsilon
$.
Since $\hat\eta \leq \etamax$, it suffices to find $\rho$ such that
$
    \rho \left( \frac{1}{\delta} + \norm{u}_2 \right) \leq \epsilon
$.
As $\norm{u}_2 \leq \|\qmax-\qmin\|_2$, setting $\rho = \frac{\delta\epsilon}{1 + \delta\norm{\qmax-\qmin}_2}$ satisfies the inequality.
\end{IEEEproof}

In the case where $\eta^\star$ is known, a similar proof shows that $k = \eta^\star + \rho \norm{u}_2$ and $\rho = \frac{\epsilon}{\|\qmax-\qmin\|_2}$ satisfy feasibility and robustness. (This can be seen as the $\delta \to \infty$ limiting case of \Cref{thm:oracle} such that consistent model chasing only updates $\hat{X}$ and keeps $\hat\eta=\eta^\star$ fixed.)

\begin{lemma}\label{lemma:tri_norm}
For all $A \in \Sym^n$, $b \in \R^n$, and $\alpha \in \R_+$,
\[
    \norm{A}_\triangle \leq \alpha
    \quad\implies\quad
    -\alpha \norm{b}_2 \one \preceq A b \preceq \alpha \norm{b}_2 \one.
\]
\end{lemma}
\begin{IEEEproof}
See \cite{yeh2022robust}.
\end{IEEEproof}

Finally, combining \Cref{thm:oracle} with \Cref{thm:ncbc} and applying \Cref{thm:ho} completes the proof of \Cref{thm:main}.

\section{Case Study}
\label{sec:experiments}

We demonstrate the effectiveness of \Cref{alg:robust_online_volt_control} using a case study based on a single-phase 56-bus network ($n=55$) from the Southern California Edison (SCE) utility, with line parameters $r_{ij}, x_{ij}$ from \cite[Table 1]{farivar2012optimal}. Even though our algorithm only has guarantees for the linear power flow model \eqref{eq:simplified-distflow}, we show that our algorithm works well on both the linear model and the more realistic nonlinear DistFlow model \eqref{eq:nonlinear-distflow}.

\begin{figure}[!tbp]
\centering
\begin{minipage}[b]{0.05\textwidth}
\includegraphics[width=\textwidth]{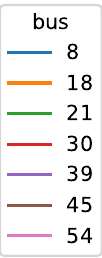}\vspace{2mm}
\end{minipage}
\subfloat[]{%
    \label{fig:linear-voltage-profile}%
    \includegraphics[width=0.4\columnwidth]{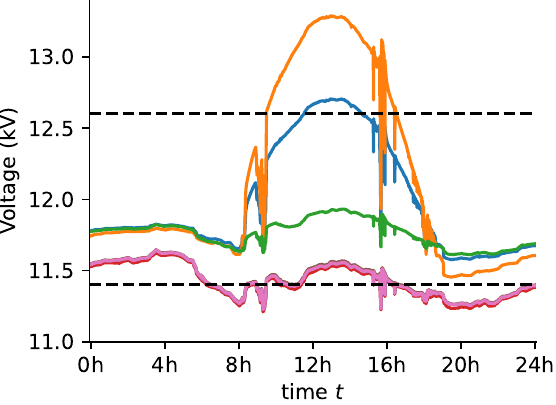}%
}%
\subfloat[]{%
    \label{fig:nonlinear-voltage-profile}%
    \includegraphics[width=0.4\columnwidth]{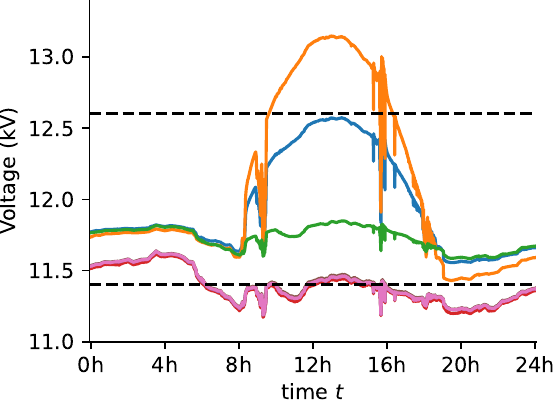}%
}
\caption{Voltage profile of 7 buses without control, simulated with \textbf{(a)} linear dynamics \eqref{eq:simplified-distflow} and \textbf{(b)} nonlinear balanced AC dynamics \eqref{eq:nonlinear-distflow}.}
\label{fig:no-control}
\end{figure}

\subsection{Experimental Setup}

Following \cite{qu2020optimal}, we adapt real-world load and PV data from \cite{bernstein2019real} for the 56-bus network by adding power injection (scaled by the PV generation) at buses $\mathcal{C} = \{2, 4, 7, 8, 9, 10, 11, 12, 13, 14, 15, 16, 19, 20, 23, 25, 26, 32\}$. Exogenous active and reactive power injection measurements are taken at each bus at 6-second intervals over a 24-hour period. \Cref{fig:no-control} plots these values for several buses to illustrate the setting considered. We assume that controllers with reactive power injection capacity are available at every node. The network parameters used in our experiments are:
\begin{itemize}
    \item nominal squared voltage magnitude at the substation \\ $v^0 = \vnom = (12\si{kV})^2$
    \item squared voltage magnitude limits \\ $\vlims = [0.95, 1.05] \text{pu} = [11.4^2, 12.6^2] \si{kV}^2$
    \item reactive power injection limits \\ $\qlims = [-0.24, 0.24] \text{MVar}$
    \item state and input cost matrices $P_v = 0.1 I$, $P_u = 10 I$
    \item initial state $v(1) = R^\star p(0) + X^\star  q^e(0) + v^0 \one$, $q^c(0) = \zero$
\end{itemize}

In comparison to previous papers in the voltage control literature, our reactive power injection limits $\qlims$ are slightly more generous than $\pm 0.2$ MVar used in, \textit{e.g.}, \cite{qu2020optimal}. We choose $\pm 0.24$ MVar because even a controller with perfect knowledge of the future would need reactive power injection capabilities of at least $\pm 0.238$ MVar in order to maintain $v(t) \in [\vmin, \vmax]$ (if $\qmin = -\qmax$) under linear dynamics \eqref{eq:simplified-distflow}.

We set $\etamax = 10$, which upper-bounds the maximum change in exogenous noise observed in our data, which is $\approx 8.6$:
\[
    \eta^\star = \max_t \norm{R^\star (p(t) - p(t-1)) + X^\star (q^e(t) - q^e(t-1))}_\infty.
\]
We fix $\epsilon=0.1$. In order to satisfy the requirement in \Cref{as:bounded_vpar} that $v(t) \in [\vmin + (\etamax+\epsilon), \vmax - (\etamax+\epsilon)]$, the reactive power injection capabilities must exceed $\pm 0.528$ MVar. As we show in experiments with only $\pm 0.24$ MVar range of control, though, \Cref{as:bounded_vpar} does not need to be fully satisfied in order for our method to still provide strong empirical results.

For the robust controller \CTRL, we set slack variable weight $\beta = 100$ and $\Vpar = [\underline\vpar, \overline\vpar]$ to be a rectangle around the true noise. Under linearized system dynamics, $\vpar(t)$ is calculated as described in \Cref{sec:volt_ctrl_background}, and then we set
\[
    \forall i \in [n]:\
    \underline\vpar_i = \min_t \vpar_i(t),\quad
    \overline\vpar_i = \max_t \vpar_i(t).
\]
Under nonlinear system dynamics, we approximate $\vpar(t)$ as the nodal squared voltage magnitudes when $q^c(t) = 0$ (as shown in \Cref{fig:no-control}), and we add $0.5\si{kV}^2$ padding which empirically suffices as a convex outer approximation of $\Vpar$:
\[
    \underline\vpar_i = \min_t \vpar_i(t) - 0.5,\quad
    \overline\vpar_i = \max_t \vpar_i(t) + 0.5.
\]

As mentioned previously, we use a greedy projection-based NCBC algorithm \cite{argue2019nearly} in \SEL{} that minimizes the movement distance $\|(\hat{X}_t,\hat\eta_t) - (\hat{X}_{t-1},\hat\eta_{t-1})\|_{\triangle,\delta}$ between nested convex sets $P_t \subseteq P_{t-1}$:
\begin{equation}\label{eq:ncbc-proj}
\begin{aligned}
    & \NCBC_\text{proj}(P_t, \hat{X}_{t-1}, \hat\eta_{t-1}) \\
    & := \argmin_{(X,\eta) \in P_t} \|(X,\eta) - (\hat{X}_{t-1},\hat\eta_{t-1})\|_{\triangle,\delta}.
\end{aligned}
\end{equation}
This achieves competitive ratio $\gamma_\text{proj}(m) = \pi (m-1) m^{m/2}$.

To keep the optimization problem \eqref{eq:cmc} computationally tractable for consistent model chasing, our implementation does not use the full trajectory $D$ as in the constraints of the consistent set \eqref{eq:consistent-set}. Instead, we include the 20 latest observations and 80 more observations sampled uniformly at random $(v(t), v(t+1), u(t), q^c(t)) \sim D$. This provides a computationally tractable approximation of the uncertainty set. In our experiments on linear system dynamics, we found that $\hat{X}_t$ selected using this approximation was always in the consistent set defined by the full trajectory $D$, when allowing for small numerical inaccuracies introduced by the CVXPY optimization solver.

Unless otherwise stated, we initialize $\hat\eta_1 = 0$. We initialize $\hat{X}_1$ by adding noise to the true $X^\star$ in two ways. First, we scale each line impedance $x_{ij}$ by a random factor $\sigma_{ij} \iid \text{Uniform}[0, 2]$. Second, we randomly permute the bus ordering, so $\hat{X}_1$ corresponds to a permuted grid topology. Finally, we project $\hat{X}_1$ into the uncertainty set $\X_\alpha$, with $\alpha=1$.

Except for the experiments shown in \Cref{fig:delta_effect}, we fix $\delta=20$ which empirically strikes a balance between minimizing the modeling error $\|\hat{X}_t - X^\star\|_\triangle$ and overfitting noise.

\begin{figure*}[!tbp]
\centering
\begin{minipage}[b]{0.04\textwidth}
\includegraphics[width=\textwidth]{figures/buses_legend.pdf}\vspace{3mm}
\end{minipage}
\subfloat[]{%
    \label{fig:linear_unknown}%
    \includegraphics[width=0.19\textwidth]{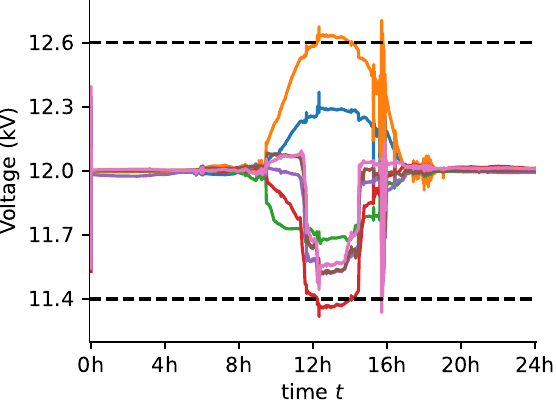}%
}%
\subfloat[]{%
    \label{fig:linear_topo-14}%
    \includegraphics[width=0.19\textwidth]{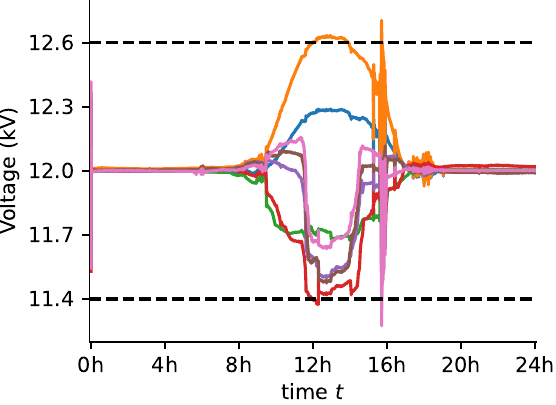}%
}%
\subfloat[]{%
    \label{fig:linear_lines-14}%
    \includegraphics[width=0.19\textwidth]{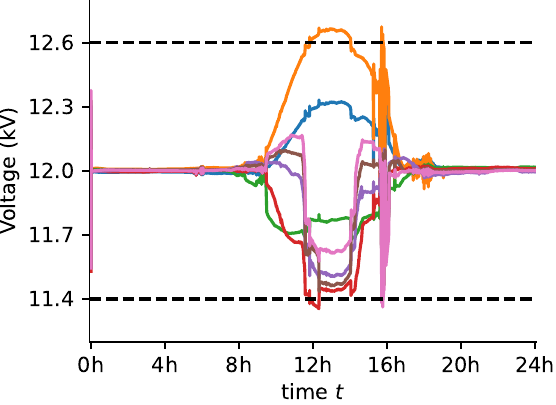}%
}%
\subfloat[]{%
    \label{fig:linear_known}%
    \includegraphics[width=0.19\textwidth]{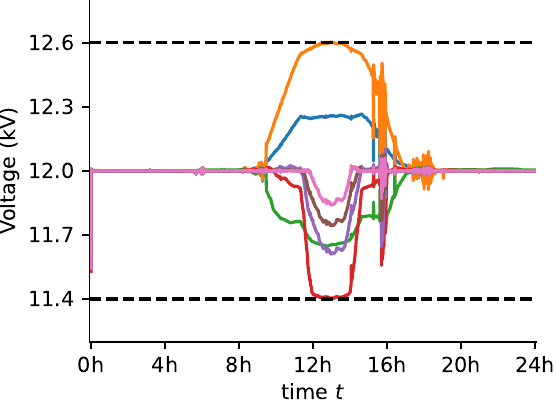}%
}%
\subfloat[]{%
    \label{fig:linear_error}%
    \includegraphics[width=0.19\textwidth]{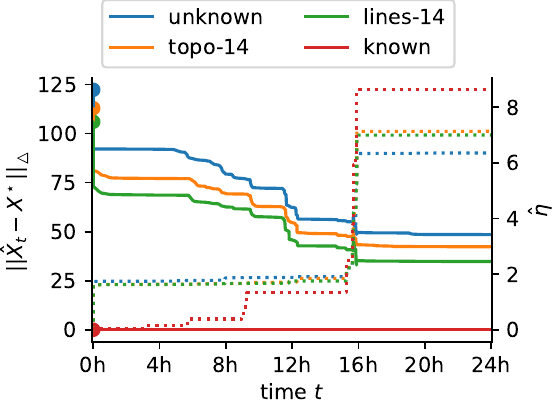}%
}
\caption{\textbf{(a)-(d)} Voltage profiles of 7 different buses simulated under linear system dynamics \eqref{eq:simplified-distflow}. Dotted black lines indicate voltage limits $\vlims$.
\textbf{(a)} \CTRL+\SEL{} initialized with random $\hat{X} \in \X_\alpha$. \textbf{(b)} like (a) but the topology for buses 1-14 is known. \textbf{(c)} like (a) but the topology and line parameters for buses 1-14 are known. \textbf{(d)} like (a) but $\hat{X}=X^\star$ is fixed and known so only $\hat\eta$ is learned \textbf{(e)} Convergence of $\hat{X}_t$ towards true $X^\star$ (solid lines, left axis) and estimated $\hat\eta$ (dotted lines, right axis). Notice that even when $\|\hat{X}_t - X^\star\|_\triangle$ does not reach 0, the controller still performs quite well.}
\label{fig:linear}
\end{figure*}

\begin{figure*}[!tbp]
\centering
\begin{minipage}[b]{0.04\textwidth}
\includegraphics[width=\textwidth]{figures/buses_legend.pdf}\vspace{3mm}
\end{minipage}
\subfloat[]{%
    \label{fig:nonlinear_unknown}%
    \includegraphics[width=0.19\textwidth]{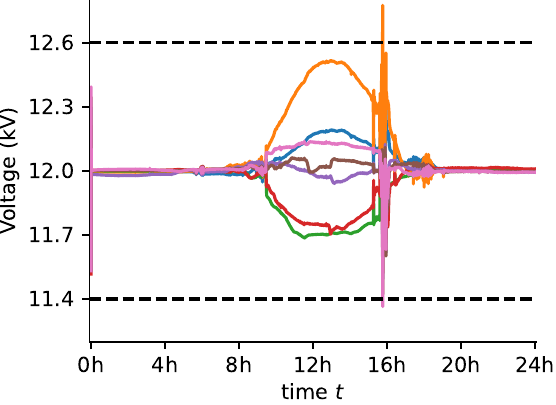}%
}%
\subfloat[]{%
    \label{fig:nonlinear_topo-14}%
    \includegraphics[width=0.19\textwidth]{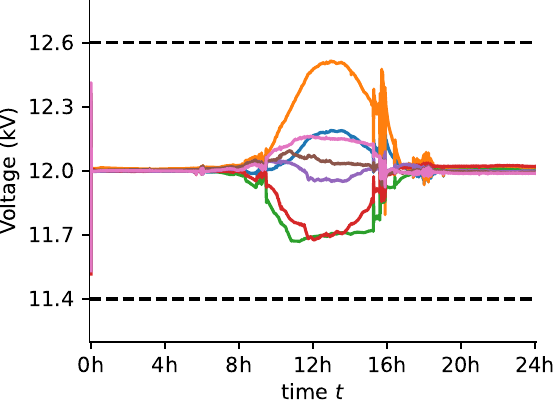}%
}%
\subfloat[]{%
    \label{fig:nonlinear_lines-14}%
    \includegraphics[width=0.19\textwidth]{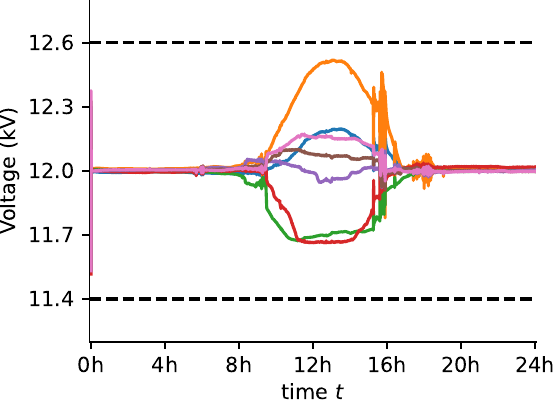}%
}%
\subfloat[]{%
    \label{fig:nonlinear_known}%
    \includegraphics[width=0.19\textwidth]{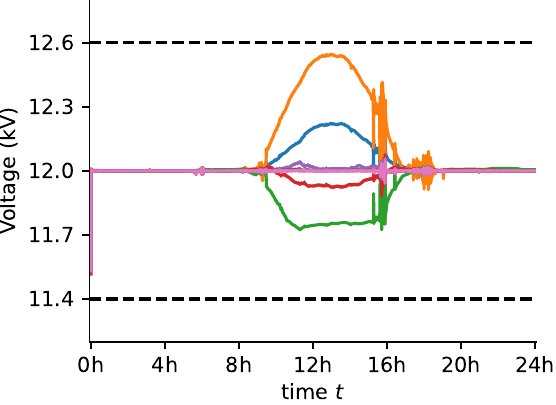}%
}%
\subfloat[]{%
    \label{fig:nonlinear_error}%
    \includegraphics[width=0.19\textwidth]{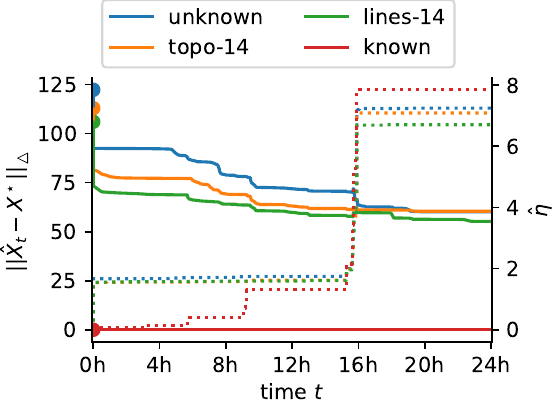}%
}
\caption{Parallels \Cref{fig:linear}. Voltage profiles of 7 different buses simulated under balanced nonlinear AC power flow \eqref{eq:nonlinear-distflow}.}
\label{fig:nonlinear}
\end{figure*}

\begin{table}[t]
\centering
\caption{Performance of our method simulated under linear system dynamics (top) and nonlinear system dynamics (bottom). See \Cref{sec:experimental-results}.}
\label{tab:quantitative-results}
\begin{tabular}{c r r r}
\toprule
     Info provided
              & \# mistakes          & avg. violation    & max violation \\
\midrule
     Unknown  &  662.2 $\pm$ 435.1 & 0.43 $\pm$ 0.16 & 4.40 $\pm$ 2.59 \\
     Topo-14  &  917.0 $\pm$ 155.2 & 0.34 $\pm$ 0.12 & 4.93 $\pm$ 2.19 \\
     Lines-14 & 1085.8 $\pm$ 186.6 & 0.57 $\pm$ 0.29 & 2.55 $\pm$ 1.09 \\
     Known    &   88.0             & 0.07            & 0.12 \\
\midrule
     Unknown  &   16.0 $\pm$  15.8 & 0.68 $\pm$ 0.56 & 2.74 $\pm$ 2.39 \\
     Topo-14  &    0.5 $\pm$   0.6 & 2.21 $\pm$ 2.56 & 2.90 $\pm$ 3.38 \\
     Lines-14 &    0.5 $\pm$   0.6 & 1.01 $\pm$ 1.20 & 1.45 $\pm$ 1.73 \\
     Known    &    0.0             & 0.00            & 0.00 \\
\bottomrule
\end{tabular}
\end{table}

\subsection{Experimental Results}
\label{sec:experimental-results}

Our experimental results demonstrate the ability of \Cref{alg:robust_online_volt_control} to stabilize the system without knowledge of the network topology, providing good voltage control performance even though it still has significant uncertainty about the topology at the end of the experiments. We test our algorithm under both the linearized system dynamics \eqref{eq:vol_dyn1} as well as the more realistic nonlinear balanced AC power flow setting \eqref{eq:nonlinear-distflow} simulated using Pandapower \cite{pandapower2018}. The convex optimization problems for \SEL{} and \CTRL{} are solved with CVXPY \cite{diamond2016cvxpy} using the MOSEK solver \cite{mosek}. Code for our simulations are available on GitHub.\footnote{\url{https://github.com/chrisyeh96/voltctrl}}

\begin{figure}
\centering
\begin{minipage}[b]{0.13\columnwidth}
\includegraphics[width=\textwidth]{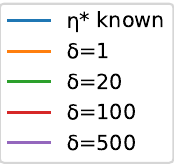}\vspace{9mm}
\end{minipage}
\subfloat[]{%
    \label{fig:delta_effect_linear}%
    \includegraphics[width=0.41\columnwidth]{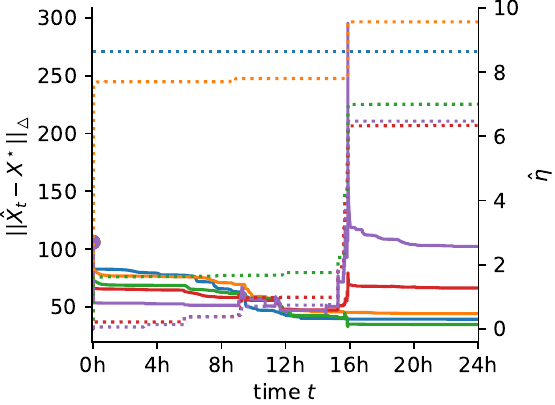}%
}%
\subfloat[]{%
    \label{fig:delta_effect_nonlinear}%
    \includegraphics[width=0.41\columnwidth]{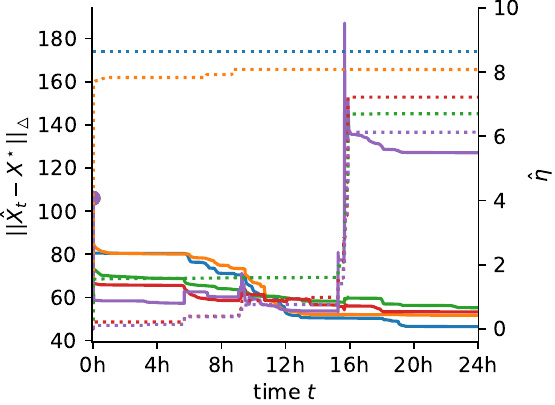}%
}
\caption{Effect of varying $\delta$ on consistent model chasing. As in \Cref{fig:linear_error}, convergence of $\hat{X}_t$ towards $X^\star$ is plotted in solid lines (left axis), and estimated $\hat\eta$ is plotted in dotted lines (right axis). In blue are results where we fix $\hat\eta = \eta^\star = 8.65$ and $\delta$ has no effect. \textbf{(a)} linear dynamics \textbf{(b)} nonlinear dynamics.}
\label{fig:delta_effect}
\end{figure}

\begin{figure}[tbp]
\centering
\begin{minipage}[b]{0.14\columnwidth}
    \includegraphics[width=\textwidth]{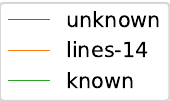}\vspace{10mm}
\end{minipage}
\subfloat[]{%
    \label{fig:partial1}%
    \includegraphics[width=0.42\columnwidth]{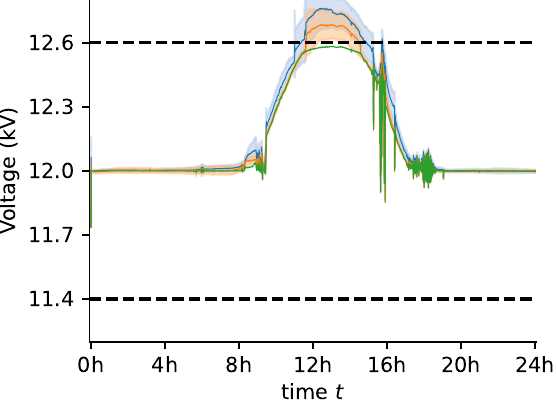}%
}%
\subfloat[]{%
    \label{fig:partial2}%
    \includegraphics[width=0.42\columnwidth]{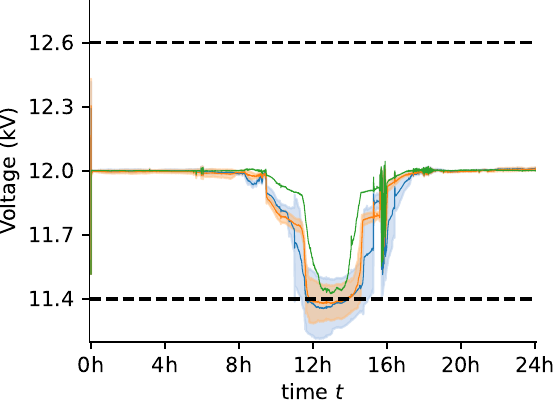}%
}
\caption{Balanced nonlinear AC power flow simulation of the voltage profiles under different algorithms with \textit{partial control and observation}. The dark colors plot the mean voltages across 4 random initializations of $\hat{X}_1$ and the light shading plots $\pm 1$ standard deviation. \textbf{(a)} bus 18 \textbf{(b)} bus 30.}
\label{fig:nonlinear_partial}
\end{figure}

\paragraph{Linearized power flow with full control}
Our first set of experiments, shown in \Cref{fig:linear} and \Cref{tab:quantitative-results} (top), tests our algorithm's performance on the SCE-56 bus network under linearized system dynamics \eqref{eq:vol_dyn1}. Different amounts of network information are provided to the consistent model chasing algorithm \SEL{} via the initial consistent set $\X_\alpha$, ranging from no information (``unknown,'' \Cref{fig:linear_unknown}), information about the edges among the first 14 buses but not the line impedances (``topo-14,'' \Cref{fig:linear_topo-14}), information about the edges \textit{and} line impedances among the first 14 buses (``lines-14,'' \Cref{fig:linear_lines-14}), and complete information about the network (``known,'' \Cref{fig:linear_known}). Because the buses in the SCE 56-bus network are numbered in a topological ordering, the ``topo-14'' setting adds constraints of the form \eqref{eq:known-topology-constraint} for all of the first 14 buses, and the ``lines-14'' setting constrains all $X \in \X_\alpha$ such that $X_{ij} = X^\star_{ij}$ for all $i,j \in \{1, \dotsc, 14\}$.

As shown in \Cref{fig:linear_error}, incorporating more prior knowledge about the network into the initial uncertainty set reduces the model estimation error $\|\hat{X} - X^\star\|_\triangle$. Furthermore, the model estimation error decreases the most dramatically when the voltage violations are the largest. However, we note that lower model estimation error does not always result in fewer mistakes in our experiments.

\Cref{tab:quantitative-results} quantifies our algorithm's performance under varying amounts of initial network information. A ``mistake'' refers to any time step where any bus' voltage violated the limits $\vlims$. ``Avg. violation'' refers to the average absolute squared-voltage violation
\[
    \mean_{i \in [n],\, t \in [T]:\ v_i(t) \not\in [\vmin_i, \vmax_i]} \max(v_i(t) - \vmax_i,\, \vmin_i - v_i(t)).
\]
``Max violation'' is like ``avg. violation'' but replaces the $\mean$ with a $\max$. Results given show the mean and standard deviation over 4 random initializations of $\hat{X}_1$.

\paragraph{Nonlinear power flow with full control}
Our second set of experiments test our online controller on the standard balanced AC power flow model \eqref{eq:nonlinear-distflow}. As in the linearized power flow experiments, we compare \Cref{alg:robust_online_volt_control}'s performance across varying levels of prior information (\Cref{fig:nonlinear} and \Cref{tab:quantitative-results}, bottom). Even though the controller is designed under the assumption of linearized voltage dynamics, our algorithm still performs well in the nonlinear simulation. The performance improves progressively, with less voltage violation and smaller overall deviation from the desired steady state voltage as it is provided more information.

\paragraph{Nonlinear power flow with partial observation and partial control}
We also test our proposed online controller in the partial observation and partial control setting. In \Cref{fig:nonlinear_partial}, we withhold voltage observations and control authority from buses $i \in \{8, 18, 21, 30, 39, 45, 54\}$ by setting $q^c_i(t)=0$ for all $t$. We simulate the voltage profiles across 4 random initializations of $\hat{X}_1$ and plot the mean and $\pm 1$ standard deviation. Despite the more challenging setting, the performance of \Cref{alg:robust_online_volt_control} remains strong. We again observe in \Cref{fig:nonlinear_partial} that adding prior topology and line parameter information marginally improves the performance of \Cref{alg:robust_online_volt_control}.

\paragraph{Varying $\delta$}
In \Cref{fig:delta_effect}, we demonstrate the effect of varying $\delta$ on the performance of our algorithm. From a theoretical perspective, \Cref{thm:main} shows that our algorithm achieves a finite mistake bound for every $\delta > 0$, and this bound is minimized by taking $\delta$ to be very large. What happens when using a large $\delta$, though, is that the model chasing algorithm may overfit to noise until a time when the noise is too large, forcing the algorithm to increase the noise bound (\textit{e.g.}, around the 16h mark in \Cref{fig:delta_effect}). This leads to inconsistent performance in the short term, albeit with perhaps better worst-case performance. In contrast, a smaller $\delta$ allows more of the network uncertainty to be captured in a larger noise $\hat\eta$ term at the cost of learning a less accurate $\hat{X}$, but the decrease in modeling error $\|\hat{X}_t - X^\star\|_\triangle$ becomes monotonic.

In practice, $\delta$ should be treated as a prior ``confidence'' about how close the initial guess of $\hat\eta$ is to $\eta^\star$. $\delta$ should be larger when there is greater confidence that $\hat\eta$ is close to the true $\eta^\star$.

\paragraph{Detecting topology changes}
Finally, we consider the challenge of responding to a change in the distribution grid topology in real-time. If the topology changes from one radial grid to another due to switches, new observed data may render the consistent set empty. That is, when consistent model chasing \eqref{eq:ncbc-proj} becomes infeasible, we are assured that the topology has changed. At this point, we may reset the algorithm by discarding the observed trajectory $D_t$ and reinitializing consistent parameter estimates from the original consistent set $P_1$. \Cref{fig:detect-topo-change} demonstrates this on linear system dynamics, where we introduce a topology change at the 12h mark. We replace lines $33 \to 40$ and $46 \to 48$ with new lines $1 \to 40$ and $10 \to 48$, which maintains a radial distribution grid.

\begin{figure}
    \centering
    \raisebox{-0.5\height}{\includegraphics[width=0.15\columnwidth]{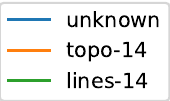}}
    \hspace{0.2cm}
    \raisebox{-0.5\height}{\includegraphics[width=0.5\columnwidth]{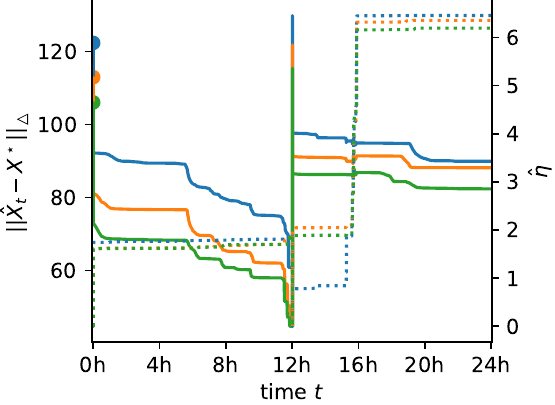}}
    \caption{Demonstration of the detection of a topology change under linear system dynamics. Convergence of $\hat{X}_t$ towards $X^\star$ is plotted in solid lines (left axis), where $X^\star$ changes at the 12h mark. The topology change triggers a reset of the consistent model chasing algorithm. Estimated $\hat\eta$ is plotted in dotted lines (right axis).}
    \label{fig:detect-topo-change}
\end{figure}
\section{Conclusion}
\label{sec:conclusion}

This paper provides the first controller that establishes a finite-mistake guarantee for voltage control in a setting with uncertainty in both the grid topology and load and generation variations. We showed that our proposed algorithm is able to learn a model of the grid dynamics in an online fashion and provably (under linearized voltage dynamics) converge to a stable controller. Further, simulated experiments on a 56-bus distribution grid demonstrate the effectiveness of our algorithm even under more realistic nonlinear dynamics. We demonstrated how to incorporate prior knowledge about the network topology and line parameters to improve performance, while also extending our algorithm to the partial observability and partial controllability setting which may better reflect real-world scenarios.

As the current algorithm is centralized, future works may consider decentralized approaches to topology-robust voltage control in order to enable faster real-time control with ideas from \cite{yu2023online}. Another direction is to extend the current algorithm to the time-varying topology setting with techniques from works such as \cite{yu2023ltv}. Further studies may also explore loosening the radial topology assumption and test our algorithm on unbalanced 3-phase AC grids to accommodate a wider range of distribution grids. This would be a challenging, but important, extension. Finally, an interesting algorithmic extension is to consider computationally efficient convex body chasing algorithms with better competitive ratios. Existing methods based on Steiner point \cite{argue2019nearly,bubeck2020chasing} achieve nearly-optimal competitive ratio but are computationally inefficient in high dimension settings such as voltage control, so designing efficient approximate Steiner point algorithms could potentially lead to significant performance improvements.

\section*{Acknowledgment}
We thank Dimitar Ho for helpful discussions about his framework.

\begin{spacing}{0.95}
\bibliographystyle{IEEEtran}
\bibliography{IEEEabrv,reference}
\end{spacing}

\end{document}